%% file: 0main.tex
\documentclass{article}
\input{head}

\title{New Direct Sum Tests}
\author{Alek Westover, Edward Yu, Kai Zheng}
\newcommand{\eps}{\varepsilon}

\begin{document}

\maketitle
\abstract{
A function $f:[n]^{d} \to \F_2$ is a \defn{direct sum} if there
are functions $L_i:[n]\to \F_2$ such that ${f(x) = \sum_{i}L_i(x_i)}$. In this work we give multiple results related to the property testing of direct sums. 

Our first result concerns a test proposed by Dinur and Golubev in \cite{di19}. We call their test the Diamond test and show that it is indeed a direct sum tester. More specifically, we show that if a function $f$ is $\eps$-far from being a direct sum function, then the Diamond test rejects $f$ with probability at least $\Omega_{n,\eps}(1)$. Even in the case of $n = 2$,
the Diamond test is, to the best of our knowledge, novel and yields a new tester for the classic property of affinity. 

Apart from the Diamond test, we also analyze a broad family of direct sum tests, which at a high level, run an arbitrary affinity test on the restriction of $f$ to a random hypercube inside of $[n]^d$. This family of tests includes the direct sum test analyzed in \cite{di19}, but does not include the Diamond test. As an application of our result, we obtain a direct sum test which works in the online adversary model of \cite{KRV}.

Finally, we also discuss a Fourier analytic interpretation of the diamond tester in the $n=2$ case, as well as prove local correction results for direct sum as conjectured by \cite{di19}.
}
% \clearpage
% \tableofcontents
\clearpage
\input{introv2}
\input{prelims}
\input{Simple-Diamond}
\input{subcube}

\input{self-correction}

\printbibliography

\end{document}

%% file: head.tex
\usepackage{graphicx} % Required for inserting images
\usepackage[left=1in, right=1in, top=1in, bottom=1in]{geometry}

\usepackage{amsthm}
\usepackage{amssymb}
\usepackage{amsmath}
\usepackage{mathtools}
\usepackage{csquotes}
\usepackage{enumitem}
\usepackage{svg}
\usepackage{hyperref}
\usepackage{xcolor}
\usepackage{fancyhdr}
\usepackage{mdframed}
\usepackage{framed}

\hypersetup{
    colorlinks,
    linkcolor={red!50!black},
    citecolor={blue!100!black},
    urlcolor={black!80!black}
}

\newcommand{\defn}[1]{\emph{\textbf{{#1}}}}
\renewcommand{\paragraph}[1]{\vspace{.5 cm} \noindent \textbf{#1} }

\usepackage{bbm}
\newcommand{\eo}{\mathsf{EvenOrOdd}}
\newcommand{\dist}{\mathsf{dist}}
\newcommand{\mc}{\mathcal}
\newcommand{\one}{\mathbbm{1}}

\DeclareMathOperator{\E}{\mathbb{E}}
\newcommand{\T}{\mathsf{T}}
\DeclareMathOperator{\Var}{\text{Var}}

\newcommand{\interior}[1]{ {\kern0pt#1}^{\mathrm{o}} }

\newcommand{\nd}{[\overline{n};d]}
\newcommand{\ind}{\mathbbm{1}}

\newcommand{\maj}{{\sf maj}}
\newcommand{\setof}[2]{\left\{ #1\; \mid \;#2 \right\}}
\newcommand{\set}[1]{\left\{ #1\right\}}

\newcommand{\F}{\mathbb{F}}

\newcommand{\N}{\mathbb{N}}

\newcommand{\ceil}[1]{\lceil #1 \rceil}

\newcommand{\ang}[1]{\langle #1 \rangle}

\usepackage{algorithm}
\usepackage[noend]{algpseudocode} % adding [noend] deletes the end while and stuff
\makeatletter
\newcounter{HALG@line}
\renewcommand{\theHALG@line}{\thealgorithm.\arabic{ALG@line}}
\makeatother

\usepackage[capitalise,nameinlink,noabbrev]{cleveref}
\crefname{equation}{}{} % cref{eq:blah} only does (1) instead of Equation (1)
\crefname{enumi}{Step}{} % cref{eq:blah} only does (1) instead of Item(1)

\theoremstyle{definition}
\newtheorem{theorem}{Theorem}[section]
\newtheorem{lemma}[theorem]{Lemma}
\newtheorem{proposition}[theorem]{Proposition}
\newtheorem{corollary}[theorem]{Corollary}

\newtheorem{fact}[theorem]{Fact}

\newtheorem{test}[theorem]{Test}

\newtheorem{claim}[theorem]{Claim}

\usepackage{thmtools}
\usepackage{thm-restate}

\usepackage[
backend=biber,
style=alphabetic,
]{biblatex}

%% file: introv2.tex
\section{Introduction}
In property testing, one is given query access to a function over a large, typically multidimensional domain, say $f: \prod_{i=1}^d S_i \xrightarrow[]{} T$. The goal is to determine whether or not $f$ satisfies some property, $\mathcal{P}$, using as few queries to $f$ as possible. Here, a property $\mathcal{P}$ can be any subset of functions. It is not hard to see that distinguishing between $f \in \mathcal{P}$ and $f \notin \mathcal{P}$ requires querying the entire domain for any nontrivial property $\mathcal{P}$, and so we typically settle for an algorithm with a weaker guarantee, typically called a property tester. A tester for $\mc{P}$ is a randomized algorithm, which given input function $f$, makes queries to $f$ and satisfies the following two properties:

\vspace{0.1cm}

\begin{itemize}
    \item \textbf{Completeness}: If $f \in \mathcal{P}$, the tester accepts with probability $1$.
    \item \textbf{Soundness}: If $f$ is $\delta$-far from $\mathcal{P}$, the tester rejects with probability $s(\delta)$.
\end{itemize}

\vspace{0.1cm}

We say that $f$ is $\delta$-far from $\mathcal{P}$, if the minimal fractional hamming distance between $f$ and any $g \in \mc{P}$ is at least $\delta$. The function $s(\delta)$ is referred to as the soundness of the tester and oftentimes one repeats the tester in $s(\delta)^{-1}$ times to reject with probability $2/3$ in the soundness case. It is clear that for any property, there is a trivial algorithm that queries $f$ on its entire domain, so oftentimes the goal is to design testers whose query complexity is sublinear in the domain size or even independent of the dimension $d$.

In this paper we consider two property testing questions related to functions known as direct sums. A function $f$ is called a direct sum if it can be written as a sum of functions on the individual coordinates, that is
\[
f(x_1,\ldots, x_d) = \sum_{i=1}^d L_i(x_i),
\]
for some functions $L_i: S_i \xrightarrow[]{} T$. Direct sums are a natural property to consider in the context of property testing. Indeed, when one takes $f: \F_2^d \to \F_2$, then the direct sum property is equivalent to the classic property of affinity (being a linear function plus a constant) considered by Blum, Luby, and Rubinfeld's seminal work \cite{BLR}. Property testing of more general versions of direct sums, called low junta-degree functions, has also been considered before in the works of \cite{BSS, amireddy2023low}. A junta-degree $t$ function is one that can be written as the sum of $t$-juntas. Thus, direct sums are junta-degree $1$ functions, and direct sum testing is a specialization of the low junta-degree  testing problem that appears in \cite{BSS, amireddy2023low}. 

In addition to being a natural property to study, motivation for direct sum testing also comes from its relation to the more widely known direct product testing problem. Historically, direct product testing was first introduced by Goldreich and Safra in \cite{GS} due to its potential for application as a hardness amplification step in more efficient PCP constructions \cite{DR, D, DM, IKW, BMV}. More recently, the related problem of direct sum testing was considered, first by \cite{david2015direct} and later by \cite{di19}. The motivation for considering direct sums stems in part from the fact that direct sums are similar to direct products in that they can be used to amplify hardness, but with the advantage that the output is a single value, as opposed to an entire tuple. This seems to make direct sum testing more difficult, but has the potential for leading to even more efficient PCPs (in particular, in improving a parameter known as the alphabet size). Analysis of a type of direct sum test is a key piece of Chan's elegant PCP construction in \cite{Chan}.

\paragraph{Direct Sum Testing}
In \cite{di19}, Dinur and Golubev give and analyze a direct sum tester which they call the Square in a Cube test. They propose an additional test which they refer to as the diamond test, but leave its analysis for future work. We call this second test the Diamond Test, and describe both of these tests below.

Let the domain be an arbitrary grid $[\overline{n}; d] = \prod_{i=1}^d [n_i]$ where each $[n_i] = \{1, \ldots, n_i\}$. We will let the output be $\F_2$ so that our input function is $$f: \nd \xrightarrow[]{} \F_2.$$ Given two points $a, b \in [\overline{n};d]$ and $x \in \{0,1\}^d$, let \defn{interpolation} $\phi_x(a,b)$ to be the vector obtained by
taking $a_i$ whenever $x_i = 0$ and $b_i$ whenever $x_i = 1$. The Square in a Cube test proceeds as follows.

\begin{test}[Square in a Cube]
    Sample random $a, b \in [\overline{n};d]$ and $x,y \in \{0,1\}^d$. Accept if and only if
    \[
    f(a) + f(\phi_x(a, b)) + f(\phi_{y}(a,b)) + f(\phi_{x+y}(a,b)) = 0.
    \]
\end{test}

The Diamond Test proceeds in a slightly different manner, which essentially fixes $x+ y = (1,\ldots, 1)$.

\begin{test}[Diamond]\label{test:diamondcube}
Sample random $a,b\in \nd,x\in \{0,1\}^d$. Accept iff
$$f(a) + f(\phi_x(a,b))+f(\phi_x(b,a)) + f(b) = 0.$$
\end{test}

In addition to the Diamond test, we also consider a family of tests that generalizes the Square in a Cube test. We call this test the Affine on Subcube test. 

\begin{test}[Affine on Subcube]\label{test:affsubcube}
Sample $a,b\in [n]^d$ and run an affinity test on the function $x\mapsto f(\phi_x(a,b))$.
\end{test}

Note that the affinity test can be any arbitrary affinity test for functions $\F_2^d \to \F_2$. The Square in a Cube test is a special case of the above with the classic BLR affinity test as the affinity test. The Diamond Test, however, is not a specialization of the above.

As is usual in property testing, it is clear that both the Diamond and Affine on Subcube test satisfy completeness, so our main interest is with regards to their soundness.
% \paragraph{This Work}

% The primary motivation of this work is to address the open question of \cite{di19} and show that the Diamond Test is a valid direct sum tester. Along the way, we also observe that the diamond test over the $n = 2$ domain is equivalent to a succinct Fourier analytic fact which we state in Section \cref{sec: fourier}.

% Apart from the Diamond Test, we observe that a more general 

% Finally, we also address a second question from \cite{di19} regarding whether a test they call the Shapka Test can be used to reconstruct the original direct sum function via a ``voting scheme''. We show that this is indeed possible provided that the given function is sufficiently close to the underlying direct sum. We then give a more query efficient voting scheme to reconstruct the underlying function. We complement these results with nearly matching bounds on the query complexity and distance beyond which recovering the underlying direct sum is impossible. 

\subsection{Main Results}
We now describe our main results. Our first result establishes soundness for the Diamond test. As the completeness is clear, this shows that the Diamond test is indeed a Direct Sum tester.

\begin{theorem} \label{3thm:diamond_cube_soundness}
Suppose $f: [n]^d \to \F_2$ passes the Diamond test with
probability $1-\eps$. Then, $f$ is $C_n\cdot \eps$-close to a direct sum, for some
constant $C_n$ independent of $d$.
\end{theorem}
We remark that the theorem above incurs a dependence on $n$. In contrast, \cite{di19} shows a version of the above result with $C_n$ replaced by an absolute constant, $C$, independent of $n$. For $n = O(1)$, the soundness we show for the Diamond Test is comparable to what is known for the Square in a Cube test, but for larger $n$ the soundness analysis becomes weaker. This dependence on $n$ in the soundness is not so uncommon in property testing results however. It is comparable to the dependence on field size in the low degree testing results of \cite{alon_lowdegree, kaufman2008algebraic}, or on grid size in the junta degree testing results of \cite{BSS, amireddy2023low}. We leave it to future work to remove the dependence on $n$ in the soundness.

Our next result is a reformulation of \cref{3thm:diamond_cube_soundness} in the $n = 2$ case. Here the Diamond Test is, to the best of our knowledge, a novel affinity tester. Moreover, we  observe that the soundness of the $n=2$ Diamond Test is equivalent to a succinct Fourier analytic fact regarding a functions distance to affinity, or equivalently its maximum magnitude Fourier coefficient. Interestingly, we were unable to show this fact without appealing to the soundness of the Diamond test.
\begin{theorem}
For $f: \{-1,1\}^d \xrightarrow[]{} \{-1,1\}$,
$$1-\frac{\max_{S \in [d]}|\hat{f}(S)|}{2} \le O(\E_{\mc{R}}[\dist(\mathcal{R}(f), \eo)]).  $$
where $\mc{R}$ is the restriction of $f$ to some randomly chosen subcube and $\eo$ is the set of functions that are either even or odd.
\end{theorem}

Along with the Diamond Test, we also describe and analyze a family of testers for direct sum which we call the Affine on Subcube test (\cref{test:affsubcube}).

\begin{theorem}
Fix $f:[n]^d\to \F_2, \eps\ge 0$ and an affinity test $T$. If $f$
passes the Affine on Subcube test (instantiated with $T$) with
probability $1-\eps$, then $f$ is $C\eps$-close to a direct sum, for some
constant $C$ independent of $n, d$.
\end{theorem}

By using a suitable erasure-resilient tester as the inner affinity tester in \cref{test:affsubcube}, we use \cref{thm:subcube} obtain direct sum testers in the online-erasure model recently introduced by \cite{KRV}. We discuss this result as well as the online-erasure model in more detail in \cref{sec:erasure}.

Finally, we address a second question posed in \cite{di19} regarding reconstructing functions which are close to
direct sums. They ask if a test that they call the Shapka test can be used to reconstruct, via a ``voting scheme'', a direct sum given query access to its corrupted version. We show that this is indeed the case, and also give an improved reconstruction method that uses fewer queries.
\begin{proposition}
Let $f: [n]^d \to \F_2$ be $\eps$ close to a direct sum $L$,
with $(n+1)\eps < 1/4$. Then $L$ can be reconstructed using a voting scheme on $f$ using $n$ queries to $f$.
\end{proposition}
 We also give lower bounds on the query complexity needed and show that our improved method is asymptotically optimal (up to constant factors) in terms of
both query complexity and the fraction of corruptions it can tolerate.

\paragraph{Proof Overview}

We give an overview of the proof of \cref{3thm:diamond_cube_soundness}. Our analysis here diverges significantly from that of Dinur and Golubev for the Square in a Cube Test. Before discussing our analysis, we first summarize their approach and where it fails for the Diamond Test.

The underlying idea behind the Square in a Cube Test is that after choosing two points $a, b \in [n]^d$, any direct sum restricted to the hypercube-like domain $\prod_{i=1}^d\{a_i, b_i\}$ must be an affine function (i.e\ multivariate polynomial of degree $1$). This idea gives way to a natural interpretation of the Square in a Cube Test

\begin{itemize}
    \item Choose two points $a, b \in [n]^d$ randomly and view the domain $\prod_{i=1}^d\{a_i, b_i\}$ as a hypercube of the appropriate dimension $d'$. Here $d'$ is the number of coordinates on which $a$ and $b$ differ.
    \item Define the function $f_{a,b}: \{0,1\}^{d'} \xrightarrow[]{} \F_2$ by $f_{a,b}(x) = f(\phi_x(a,b))$.
    \item Perform the BLR affinity test on $f_{a,b}$. That is, choose $x,y,z \in \{0,1\}^{d'}$ and accept if and only if 
    \[
    f_{a,b}(x) + f_{a,b}(x+y) + f_{a,b}(x+z) + f_{a,b}(x+y+z) = 0.
    \]
\end{itemize}
From here, the analysis of \cite{di19} has two parts. First, they apply BLR to show that if the overall Diamond in a Cube test passes with high probability, then for many $a, b$, the function $f_{a,b}$ is close to an affine function (or direct sum), say $F_{a,b}$, on the hypercube $\prod_{i=1}^d\{a_i, b_i\}$. These $F_{a,b}$ can be thought of as local direct sums that agree with $f$ locally on hypercubes inside of $[n]^d$. The second step is to show that many of these $F_{a,b}$ are actually consistent with eachother and apply a direct product testing result of \cite{dinur2014direct} to conclude that there is in fact a global direct sum $F$ consistent with many $F_{a,b}$'s. As the $F_{a,b}$'s are in turn largely consistent with $f$, we get that $F$ is a direct sum close to $f$, concluding the proof of soundness. 

Our analysis of the soundness of the Diamond in the Cube test, however, requires a different approach. Let us first see what happens when we try to naively adapt the above strategy. After choosing $a, b \in [n]^d$, instead of applying BLR, the Diamond in the Cube test performs the following test on $f_{a,b}: \{0,1\}^{d'} \xrightarrow{} \F_2$. 
\[
\text{Choose $x \in \{0,1\}^{d'}$ and accept if and only if $f_{a,b}(0) + f_{a,b}(x) + f_{a,b}(x+1) + f_{a,b}(1) = 0$.}
\]
There is an issue though: the above is not an affinity tester! Indeed one can check that if $f_{a,b}(0) + f_{a,b}(1) = 0$, resp.\ $1$, then any even, resp.\ odd, function passes the above test with probability $1$. Thus, we would only get that many $f_{a,b}$ are close to even or odd functions, and from here it is not clear how to make the second part of the analysis go through. 

We instead combine ideas from the soundness analyses of \cite{david2015direct, optimalreedmuller, amireddy2023low}, which use induction. Suppose $f: [n]^d \xrightarrow[]{} \F_2$ is $\eps$-far from a direct sum. We wish to show that the Diamond Test rejects $f$ with some non-trivial probability that is independent of $d$. The inductive approach of both of the mentioned works consist of three main steps: 1) Restate the test on $f$ as choosing a random subdomain that is one dimension and performing the test on $f$ restricted to this subdomain, 
2) analyze the soundness in the $\eps$ very small case, 3) on the other hand, show that if $\eps$ is large then show that $f$ must also be far from a direct sum on many subdomains. 

To execute all of the above steps, we have to draw on ideas from all of the mentioned works. For step 1), we 
take inspiration from \cite{david2015direct} and instead analyze a four-function version of the Diamond Test. The reason is that there is no clear way to view the Diamond Test in the manner described above. It is tempting to describe the Diamond Test as choosing a random coordinate, say $i = 1$, and a random value in $[n]$, say $1$, and then performing a $d-1$-dimensional version of the Diamond Test on the function $f(1, x_2, \ldots, x_d)$. Unfortunately this does not work as one has to change the distribution that the points in the $d-1$-dimensional grid are chosen. Indeed, if one chooses the points $a,b \in [n]^{d-1}$ of the one dimension down Diamond Test uniformly at random, then the final $d$-dimensional points output would agree on $1 + (d-1)/n$ coordinates in expectation instead of $d/n$. 

For the small distance case in step 2), we rely on the a hypercontractivity theorem over grids. However, since we are now working with a four-function version of the Diamond Test some additional care is needed in this analysis.

Finally for step 3), we use techniques from agreement testing in the PCP literature \cite{RS, MZpcp}. From \cite{RS} we borrow their transitive-consistency graph technique used to analyze the so-called Plane versus Plane test and from \cite{MZpcp} we show a version of their hyperplane sampling lemma for $(d-1)$-dimensional grids in $[n]^d$.

\paragraph{Paper Outline}
In \cref{sec:diamond} we present and analyze our novel affinity tester (the Diamond test).
In \cref{sec:subcube} we give a new class of direct sum tests. In \cref{sec:correction} we consider the problem of obtaining ``corrected'' samples from corrupted direct sums.

%% file: prelims.tex
\section{Preliminaries} \label{sec:prelims}
We now introduce some notations that will be used throughout the paper. We let $[n]$ denote $\set{1,2,\ldots, n}$. Let $[n;d]$ denote a product space $[n_1]\times \cdots \times [n_d]$.
For a set or event $X$, we write $\overline{X}$  to denote the complement of $X$.

The notion of distance that we use is fractional hamming distance. That is, given two functions $f: S \xrightarrow[]{} T$ and $g: S \xrightarrow[]{} T$, the distance between them is the fraction of entries on which they differ:
% @KAI: I STRONGLY prefer that we use dist instead of delta for distance. because we use delta all over the place. 
\[
\dist(f, g) = \Pr_{x \in S}[f(x) \neq g(x)].
\]
Given a property $\mathcal{P} \subseteq \{g: S \xrightarrow[]{} T\}$, the distance from $f$ to the property $\mc{P}$ is defined as
\[
\dist(f, \mc{P}) = \min_{g \in \mathcal{P}} \dist(f, g).
\]
We say that $f$ is \emph{$\delta$-far} from $\mathcal{P}$ if $\dist(f, \mc{P}) \geq \delta$.

Finally, if we do not explicitly specify the distribution a random variable is drawn from
is a probability, then the random variable is meant to be drawn uniformly from
the appropriate space (which will be clear from context).
The constants hidden in our $O$-notation are never allowed to depend on $d$.
If there is a parameter $k$ that we think of as constant, we will sometimes use
the notation $O_k(\cdot)$ to emphasize that the constant hidden by the $O$ is
allowed to depend on $k$.

\paragraph{Boolean Functions}
Define the \defn{interpolation} $\phi_x(a,b)$ to be the vector obtained by
taking $a_i$ whenever $x_i = 0$ and $b_i$ whenever $x_i = 1$. 
A frequently useful property of interpolations is:
\begin{fact}\label{fact:fliperoo} 
$ \phi_x(a,b)=\phi_{x+\one}(b,a). $
\end{fact}
Given $x\in \prod_i S_i$ and $p\in [0,1]$, let $y\sim \T_p(x)$ denote a vector in $\prod_{i}
S_i$, with each $y_i$ sampled independently as follows: set $y_i = x_i$ with
probability $p$, and otherwise sample $y_i$ uniformly from $S_i$; $\T_p$ is
called the \defn{noise operator}. Given $p\in [-1,0]$ and $x\in \F_2^{d}$, we
define the distribution $y\sim \T_p(x)$ as: $y_i=x_i+1$ with probability $-p$
and $y_i$ is sampled uniformly from $\set{0,1}$ otherwise.
For a boolean function $f:\F_2^d\to \F_2$ we will use $\mu(f)$ to denote the
fractional size of the support of $f$ (i.e., the fraction of $\F_2^d$ on which
$f$ outputs $1$). We use $\one$ to denote the all ones vector, and sometimes use
$0$ to denote the zero-vector (the dimension will be clear from context). For
$v,w\in \F_2^d$ we write $\ang{v,w}$ to denote $\sum_i v_i w_i.$ For boolean
functions $f,g:\F_2^d\to \F_2$ we write $\ang{f,g}$ to denote $\E_x[f(x)g(x)]$.

When discussing functions $f:[n]^d\to \F_2$ we will always assume $n>1$.

%% file: simple-diamond.tex
\section{The Diamond Test}\label{sec:diamond}
We restate the Diamond test, initially proposed by Dinur and Golubev in \cite{di19} (recall from \cref{sec:prelims} that $\phi_x$ is the interpolation function)
\begin{test}[Diamond]\label{3test:diamondcube}
Sample random $a,b\sim [n]^d,x\sim \F_2^d$. Accept iff
$$f(a)+f(b)=f(\phi_x(a,b))+f(\phi_x(b,a)).$$
\end{test}
% We note that, to the best of the authors' knowledge, this test was not
% previously understood in the setting $n=2$. That is, an analysis of the Diamond
% in Cube test for the case $n=2$ gives a novel affinity tester for functions
% $f:\F_2^d\to \F_2$. This is in contrast to Dinur and Golubev's Square in Cube
% test for direct sum, which reduces to the BLR \cite{BLR} tester in the case of % $n=2$.
We will show that \cref{3test:diamondcube} is a valid tester for direct sum. 
First, we analyze the case that the test passes with probability $1$. 
\begin{lemma}\label{3lem:basecase}
The function
$f: [n]^d \to \F_2$ is a direct sum if and only if $f$ passes the Diamond test with probability $1$.
\end{lemma}
\begin{proof}
If $f$ is a direct sum then we can write $f(x) = \sum_i L_i(x_i)$. It follows that for any $a,b \in [n]^d$,
$$f(a)+f(b)  = \sum_i (L_i(a_i)+L_i(b_i)) = f(\phi_x(a,b)) + f(\phi_x(b,a)).$$
For the other direction, suppose $f$ passes the Diamond test with probability
$1$. Then, for all $x\in \F_2^d, a\in [n]^d$ we have
\begin{equation}\label{3eq:breakitup}
f(a) = f(\one)+f(\phi_x(\one,a)) + f(\phi_x(a,\one)).
\end{equation}
Let $e_{i,j}\in [n]^d$ denote a vector which is
$1$ at all coordinates except coordinate $i$, where it has value $j$.
Then, repeated application of \cref{3eq:breakitup} lets us decompose $f(a)$ as 
$$f(a) = f(\one)+\sum_{i=1}^d \left(f(e_{i,a_i})+f(\one)\right).$$
Thus, $f$ is a direct sum.
\end{proof}
As a simple corollary of \cref{3lem:basecase} we have:
\begin{corollary}
If $f$ is $\eps$-close to a direct sum, then $f$ passes the Diamond test with probability at least $1-4\eps$.
\end{corollary}
\begin{proof}
Let $g$ be a direct sum with $\dist(f,g)\le \eps$. By a union bound, $f$ and $g$ agree
on all $4$  points queried by the test (because each query is marginally
uniformly distributed), in which case the Diamond  test accepts.
\end{proof}

With completeness done, the remainder of this section is focused on showing the Diamond test's soundness as stated in \cref{3thm:diamond_cube_soundness}. We restate the theorem below.

\begin{theorem} 
Suppose $f: [n]^d \to \F_2$ passes the Diamond test with
probability $1-\eps$. Then, $f$ is $C_n\eps$-close to a direct sum, for some
constant $C_n$ independent of $d$.
\end{theorem}

We sketch the ideas used in the analysis. First, we will
generalize to a ``4-function'' version of the test, which is more amenable to
induction. Then, we use hypercontractivity (\cite{hyper1},\cite{hyper2}, \cite{amireddy2023low}) to
prove \cref{3thm:dichotomy}, which establishes a dichotomy for 4-tuples of
functions that pass the test with probability $1-\eps$: if $(f,g,k,h)$ passes
the test with probability $\eps$, then either $f$ (and $g,h,k$ as well) is
$2\eps$-close to a direct sum, or $f$ is $\Omega_n(1)$-far from a direct
sum.
Next, we will relate the pass probability of $(f,g,h,k)$ to the pass probability
of 1-variable restrictions of $f,g,k,h$. This allows us to go down one dimension
and apply induction. Crucially, the dichotomy theorem allows us to avoid our
closeness factor deteriorating at each step of the induction, which would result
in a dependence on the dimension in our result. This dichotomy based approach is
inspired by the work of \cite{optimalreedmuller, david2015direct}, but the
details of establishing the dichotomy and performing the induction are quite
different.

\subsection{The Four Function Diamond Test}
In order to facilitate our inductive approach we define a four function version of the diamond test. 
\begin{test}[4-Function Diamond]
Sample $x\sim \F_2^d$ and $a,b\sim [n]^d$. Accept iff
\[ f(a) + g(\phi_x(a,b)) + h(\phi_x(b,a)) + k(b) = 0. \] 
\end{test}
The following lemma characterizes what happens when four functions pass with probability $1$. It will be used later to establish the base case (in constant dimensions) of our inductive analysis.
\begin{lemma}\label{3lem:4basecase}
Suppose $(f,g,h,k)$ passes the $4$-function Diamond test with probability $1$.
Then $f,g,k,h$ are direct sums.
\end{lemma}
\begin{proof}
If $(f,g,h,k)$ passes the 4-function Diamond test with probability $1$ then for all $a,x,b$ we have:
\[ f(a)+g(\phi_x(a,b))+h(\phi_x(b,a))+k(b)=0=f(b)+g(\phi_{x+\one}(b,a))+h(\phi_{x+\one}(a,b))+k(a). \] 
Using \cref{fact:fliperoo} to simplify we get that for all $a,b$,
\[ f(a)+k(a)=f(b)+k(b). \] 
Again using the fact that $(f,g,h,k)$ passes the test with probability $1$ we have that for all $a,x,b$:
\[ f(a)+g(\phi_x(a,b))+h(\phi_x(b,a))+k(b)=0=f(a)+g(\phi_{x+\one}(a,b))+h(\phi_{x+\one}(b,a))+k(b). \] 
By \cref{fact:fliperoo} this implies
\[ g(\phi_x(a,b))+g(\phi_{x}(b,a)) = h(\phi_x(a,b))+h(\phi_{x}(b,a)). \] 
Letting $c=\phi_x(a,b),d=\phi_{x}(b,a)$, and noting that $c,d$ can be
arbitrary (unrelated) elements of $[n]^{d}$, we have that for all $c,d$.
\[ g(c)+g(d)  = h(c)+h(d).\] 
Combining our observations so far we conclude that there exist constants
$\alpha,\beta$ such that for all $a\in [n]^{d}$ we have
\[ f(a)=k(a) +\alpha, \quad g(a) = h(a)+\beta.\] 
Thus, for all $a,b$ we have
\[ f(a)+f(b)=g(\phi_x(a,b))+g(\phi_x(b,a))+\alpha+\beta. \] 
Taking $x=0$ we conclude
\[ f(a)+f(b) = g(a)+g(b)+\alpha+\beta. \] 
Thus, there exists $\gamma$ such that  for  all $a$ we have
\[ f(a)=g(a)+\gamma.\] 
In summary, we have shown that for all $a,b,x$, we have
\[ f(a)+f(b)=f(\phi_x(a,b))+f(\phi_x(b,a))+\alpha+\beta. \] 
Setting $x=0$ again shows that $\alpha+\beta = 0$.
Thus, we have that $f$ passes the (1-function) Diamond test with probability $1$.
By \cref{3lem:basecase} this implies that $f$ is a direct sum.
Finally, recall that we have shown in the course of our proof that $g,h,k$ all
differ from $f$ by a constant. This concludes the proof.
\end{proof}

\subsection{A Dichotomy for Functions Passing the Diamond Test}
We are now prepared to show \cref{3thm:dichotomy}, which states
that if the 4-function Diamond test accepts $(f,g,h,k)$ with probability exactly
$1-\eps$ then $f$ is either $2\eps$-close to a direct sum, or $\Omega_n(1)$-far from being a direct sum. 
To this end, we first establish two lemmas which will useful in
the proof.

\begin{lemma}\label{3lem:fghconst}
Fix $d\in \N$ and functions $f,g,h:[n]^{d}\to \F_2$.
Suppose $f,g,h$ have
\[ \Pr_{a,b\sim [n]^d,x\sim \F_2^d}[f(a)+g(\phi_x(a,b))+h(\phi_x(b,a))=1]<\eps. \] 
Then, $f,g,h$ are $3\eps$-close to constants.
\end{lemma}
\begin{proof}
Fix $f,g,h$ as in the theorem. 
Performing a union bound and using \cref{fact:fliperoo} we have
\[ \Pr_{a,b,x}[f(a)+g(\phi_x(a,b))+h(\phi_x(b,a))=f(b)+g(\phi_x(a,b))+h(\phi_x(b,a))] >  1-2\eps.\] 
Thus, $f$ is $2\eps$-close to a constant $\alpha$.
Then, by a union bound we have
\[ \Pr_{a,b,x}[g(\phi_x(a,b))+h(\phi_x(b,a))=\alpha+1]>1-3\eps. \] 
Of course, the distribution of $(\phi_x(a,b),\phi_x(b,a))$ is the same as the
distribution of $(a,b)$.
Thus, we have
\[ \Pr_{a,b}[g(a)+h(b)= \alpha+1]>1-3\eps. \] 
By the probabilistic method this implies that $g,h$ are each $3\eps$-close to
constants.
\end{proof}

We will also need the following lemma, which is an extension of the classic
\defn{small set expansion} property of the noisy hypercube (see \cite{od21}) to \emph{grids}.
\begin{fact}[\cite{od21}] \label{fact:grid-hyper}
% [Theorem 6.4 of \cite{amireddy2023low}]
% TODO: this isn't actually exactly the same as their theorem... maybe worth putting in a sentence about why their theorem implies this fact.
Fix $n\in\N$. 
There exists $\lambda_n > 0$ such that for any $d$, and any set $S\subseteq [n]^d$ we have
\[\Pr_{a\sim [n]^d, b\sim \T_{1/2}(a)}[a\in S \land b\in S]\le 2\mu(S)^{1+\lambda_n}.\] 
\end{fact}

Now we are ready to establish the dichotomy theorem.
\begin{theorem}\label{3thm:dichotomy}
Fix $n\in \N$. There exists a constant $c_n> 0$ such that the following is
true for any $d$ and any functions $f,g,h,k:[n]^{d}\to\F_2$.
Suppose $k$ has $\dist(k, \mathsf{direct sum}) < c_n$.
Then, the $4$-function Diamond test rejects $(f,g,h,k)$ with probability at
least $\dist(k,\mathsf{direct sum})/2$.
\end{theorem}
\begin{proof}
Fix $f,g,k,h$ as described in the theorem statement.
Let $\chi$ be the closest direct sum to $k$. Then,
\[ \chi(a)+\chi(b)=\chi(\phi_x(a,b))+\chi(\phi_x(b,a)). \] 
Thus, our goal is equivalent to showing that
\[ \Pr_{a,x,b}[(\chi+f)(a)+(\chi+g)(\phi_x(a,b))+(\chi+h)(\phi_x(b,a))\neq
(\chi+k)(b)]\ge \dist(k,\mathsf{direct sum})/2. \] 
Thus, it suffices to consider the case that the closest direct sum to $k$
is the zero function. We assume this for the remainder of the proof. Then,
$\dist(k,\mathsf{direct sum})=\mu(k)$.
We are already done unless
\begin{equation}
  \label{3eq:eatzero}
\Pr_{a,x,b}[f(a)+g(\phi_x(a,b))+h(\phi_x(b,a))=1]\le 2\mu(k).
\end{equation}
Thus, we may assume that \cref{3eq:eatzero} is true.
Now, by \cref{3lem:fghconst} we have that $f,g,h$ are $6\mu(k)$-close to
constants.
Without loss of generality we may assume $g,h$ are both closer to $0$ than to
$1$; otherwise we can add $1$ to both $g,h$, which does not affect the sum
$g(\phi_x(a,b))+h(\phi_x(b,a))$.
If $f$ were closer to $1$ than to $0$ then we would have 
\[ \Pr_{(a,x,b)}[f(a)=1,g(\phi_x(a,b))=h(\phi_x(b,a))=0]\ge 1-18\mu(k) > 2\mu(k), \] 
with the final inequality holding by our assumption that $\mu(k)<c_n$ (we will
choose $c_n<.001$); but this
would contradict \cref{3eq:eatzero}, so it cannot happen.
Thus, $f,g,h$ are all $6\mu(k)$-close to $0$.

Now, we return to analyzing the probability that the 4-function Diamond test
rejects $(f,g,h,k)$; we write $\Pr[\mathsf{rej}]$ to denote the probability. 
Then,
\begin{align*}
\Pr[\mathsf{rej}]&\ge \Pr_{a,x,b}[k(b)=1 \land f(a)=g(\phi_x(a,b))=h(\phi_x(b,a))=0]\\
&\ge \mu(k)-\Pr_{a,b}[k(b)=f(a)=1]-\Pr_{a,x,b}[k(b)=g(\phi_x(a,b))=1]-\Pr_{a,x,b}[k(b)=h(\phi_x(b,a))=1]\\
&\ge \mu(k)-\mu(k)\mu(f)-\Pr_{a,b\sim\T_{1/2}(a)}[k(b)=g(a)=1]-\Pr_{a,b\sim\T_{1/2}(a)}[k(b)=h(a)=1].
\end{align*}
Then, using \cref{fact:grid-hyper} we have
\[\Pr[\mathsf{rej}] \ge \mu(k)(1-\mu(f)) - (\mu(k)+\mu(g))^{1+\lambda_n} - (\mu(k)+\mu(h))^{1+\lambda_n}),\]
where $\lambda_n>0$ is the constant from \cref{fact:grid-hyper}.
Recalling that $\mu(f),\mu(g),\mu(h)\le 6\mu(k)$ we have: 
\[\Pr[\mathsf{rej}] \ge \mu(k)(1-\mu(k)-100\mu(k)^{\lambda_n}).\]
Thus, there is some constant $c_n$ (dependent
only on $n$) such that if $\mu(k)\le c_n$ then the rejection probability in
is at least $\mu(k)/2$, as desired.
\end{proof}

\input{restriction_lemma}

\subsection{Soundness of the Diamond Test}
Equipped with \cref{3thm:dichotomy} we complete the proof of \cref{3thm:diamond_cube_soundness}. Ultimately, we wish to conclude that if the 4-function Diamond test rejects $(f,g,h,k)$ with probability exactly $\eps$, then $f,g,h,k$ are all $O(\eps)$-close to direct sums.
However, what we have thus far only allows us to tonclude that $f,g,k,h$ are each \emph{either} $c_n$-far from being a direct sum, \emph{or} $2\eps$-close to a direct sum.
Our next theorem rules out the former possibility.
\begin{theorem}\label{3thm:induct}
Fix $n\in \N$.
Let $c_n$ be the constant from \cref{3thm:dichotomy}.
There is a constant $\delta_n$ such that the following holds.
Fix $d\in \N$, and functions $f,g,h,k: [n]^{d}\to \F_2$. Suppose $f$ is
$c_n$-far from being a direct sum. Then, the 4-function Diamond test rejects $(f,g,h,k)$
with probability at least $\delta_n$.
\end{theorem}

\begin{proof}
% \begin{proof}[Proof of \cref{3thm:induct}]
By \cref{lm: restriction} there exists a constant $K_n$ such that if $f$ has at least $100n^2/c_n$ restrictions which are $c_n/K_n$-close to affine, then $f$ is $c_n$-close to affine.
Define $N_n = \ceil{100n^2/c_n}+1$ and $\delta_n = \min(c_n/(2K_n), n^{-3N_n})$.
We will prove the theorem by induction on the dimension $d$.

The base case is $d\le N_n$.
If $d\le N_n$ then by \cref{3lem:basecase} the 4-function Diamond test rejects
$(f,g,h,k)$ with non-zero probability since $f$ is not a direct sum. But, any non-zero
probability over events determined by $a,b,x$ is at least $1/n^{3N_n}\ge
\delta_n$.

Now, fix $d> N_n$, and assume the theorem for functions on $\F_2^{d-1}$.
Let $f\mid_{i,\sigma}:[n]^{d-1}\to\F_2$ denote the function obtained by taking $f$ and fixing $x_i=\sigma$.
Then, there is some $i\in [d]$ such for all $\sigma\in [n]$ the restriction $f\mid_{i,\sigma}$ is
$c_n/K_n$-far from a direct sum.
To analyze the probability that the 4-function Diamond test rejects $(f,g,h,k)$,
we consider a partition of the probability space into $2n^2$ disjoint events. 
The events are the value in $[n]^2\times \F_2$ that our sampled $(a_i,b_i,x_i)$ takes.
Now, consider the probability that the 4-function Diamond test rejects
$(f,g,h,k)$, conditional on $(a_i,b_i,x_i)=(\alpha,\beta,\xi)$ for some $\alpha,\beta,\xi$.
Conditional on this event, we will sample $a,b\in[n]^{d-1},x\in \F_2^{d-1}$ and check:
\[ f\mid_{i,\alpha}(a)+g\mid_{i,\phi_\xi(\alpha,\beta)}(\phi_x(a,b)) =
h\mid_{i,\phi_\xi(\beta,\alpha)}(\phi_x(b,a))+k\mid_{i,\beta}(b). \] 
This is precisely a 4-function Diamond test on $(d-1)$-dimensional functions.
Now we consider two cases.\\
\textbf{Case 1: $f\mid_{i,\alpha}$ is $c_n$-far from being a direct sum}. \\
Then by the inductive hypothesis this $(d-1)$-dimension test will reject with
probability at least $\delta_n$.\\
\textbf{Case 2:} $\dist(f\mid_{i,\alpha}, \mathsf{direct sum})\in [c_n/K_n, c_n]$.\\
Then, by \cref{3thm:dichotomy} this $(d-1)$-dimensional test rejects with
probability at least $c_n/(2K_n)\ge \delta_n$.

We have seen that in all cases, the lower dimensional 4-function Diamond tests
reject with probability at least $\delta_n$.
Thus, the actual 4-function Diamond test rejects with probability $\delta_n$ as well.
\end{proof}

Combining \cref{3thm:dichotomy} and \cref{3thm:induct} we have get the following corollary and conclude the proof of \cref{3thm:diamond_cube_soundness}.
\begin{corollary}\label{3cor:4fun}
If $(f,g,h,k)$ passes the 4-function Diamond test with probability
$1-\eps$, then $f,g,h,k$ are each $O_n(\eps)$-close to direct sums. 
\end{corollary}
\begin{proof}
    The result follows immediately from  \cref{3thm:dichotomy} and \cref{3thm:induct}.
\end{proof}
We can now conclude the proof of \cref{3thm:diamond_cube_soundness} and establish soundness for the Diamond test.

\begin{proof}[Proof of \cref{3thm:diamond_cube_soundness}]
    The result follows from \cref{3cor:4fun} by setting $f = g = h = k$.
\end{proof}

\input{fourierfact}

%% file: restriction_lemma.tex
\subsection{A Lemma Concerning Restrictions}
To complete our analysis of the Diamond test, we will relate the Diamond test to lower dimensional Diamond tests, and argue inductively. The Dichotomy lemma just established will be the key to ensuring that our soundness parameter does not deteriorate as we perform the induction. First, we need a lemma that controls the restrictions of a function which is far from a direct sum; namely, we show that a function which is far from a direct sum cannot have too many restrictions which are close to direct sums.

For some set $R \subseteq [n]^d$, we let $f|_R$ denote the restriction of $f$ to $R$, so that $f|_R: S \xrightarrow[]{} \F_2$ is given by $f|_R(x) = f(x)$ for all $x \in R$. We will specifically be considering restrictions where $R$ is of the form $R = \{x \in [n]^d  \; | \; x_i = j\}$ for some $1 \leq i \leq d$ and some $j \in [n]$. We call such $R$, \emph{one-variable restrictions}. Throughout this subsection we use $\ind$ to denote the indicator of an event. The goal of this subsection is to establish the following Lemma.

\begin{lemma} \label{lm: restriction}
    Suppose there are $100n^2/\eps$ one-variable restrictions $R$ such that $f|_R$ is $\varepsilon$-close to a direct sum on $R$. Then, $f$ is $O(\varepsilon)$-close to a direct sum.
\end{lemma}

We will need Dinur and Golubev's result regarding the soundness of the square in the cube test, which we state below.
\begin{theorem} \cite{di19} \label{thm: sq in cu}
Let $f: [n]^d \xrightarrow{} \F_2$ satisfy,
\[
\Pr_{a,b \in [n]^d, x,y,z \in \{0,1\}^d}[f(a) + f(\phi_x(a,b)) + f(\phi_{y}(a,b)) + f(\phi_{x+y}(a,b)) = 0] \geq 1-\eps.
\]
Then $f$ is $O(\eps)$-close to a direct sum.
\end{theorem}

At a high level our proof of \cref{lm: restriction} follows the strategy of the plane versus point analysis of \cite{RS}. Suppose that there are many restrictions $R_i$ with local direct sums $f_i: R_i \xrightarrow[]{} \F_2$ such that $f_i$ and $f$ are close. We will first show that many --- in fact at least a constant fraction --- of these $f_i$'s are actually consistent with each other. Using these consistent $f_i$'s we can then define a function $F$ on the whole space and argue that 1) $F$ is close to $f$ and 2) $F$ passes the square in the cube test with high probability and is thus close to a direct sum. For both 1) and 2) we are crucially using the fact that $F$ is defined using many consistent direct sums, so for nearly every point $x$ in the domain, there are in fact many of these local direct sums defined on $x$. 

\subsubsection{A Sampling Lemma}
In order to carry out the mentioned strategy, we must first show a sampling lemma. This lemma roughly states that given a large enough set of restrictions $R_i$, the measure on $[n]^d$ produced by first choosing a random $R_i$ and then choosing a random $x \in R_i$ is nearly equal to the uniform measure up to constant factors (which will be negligible). 

Let $\mathcal{R}$ be a set of one variable restrictions and let $M = |\mathcal{R}|$. Define the measure $\nu_{\mc{R}}$ on points $[n]^d$ via the following sampling procedure:
\begin{itemize}
    \item Choose $R \in \mathcal{R}$.
    \item Choose $x \in R$.
\end{itemize}
Letting $N_{\mc{R}}(x) = |\{R \in \mc{R} \; | \; R \ni x\}|$, it is clear that, 
\[
\nu_{\mc{R}}(x) = \frac{N_{\mc{R}}(x)}{M} \cdot \frac{1}{n^{d-1}}.
\]
For a subset of points $S \subseteq [n]^d$, we define
\[
\nu_{\mc{R}}(S) = \sum_{x \in S}\nu_{\mc{R}}(x).
\]
We also let $\mu$ denote the standard uniform measure on $[n]^d$, so for a set of points $S \subseteq [n]^d$, we have $\mu(S) = \frac{|S|}{n^d}$. At times we will also use $\mu$ to denote the uniform measure over grids of dimension other than $d$. It will be clear from context when this is the case and we remark that one should think of the dimension in this subsection as being arbitrary. 

\begin{lemma} \label{lm: var bound}
   Let $\mc{R}$ be a set of one-variable restrictions of size $M$. Then,
    \[
    \E_{x}[N_{\mc{R}}(x)] = \frac{M}{n}\quad \text{and} \quad \Var(N_{\mc{R}}(x)) \leq \frac{M}{n}.
    \]
\end{lemma}
\begin{proof} 
   For the first part, we have, by linearity of expectation:
    \[
    \E_{x}[N_{\mc{R}}] = \sum_{R \in \mc{R}}\E_{x}[\ind(x \in R)] = \frac{M}{n}.
    \]
    Towards the second part, we write
    \begin{align*}   
    \E[N_{\mc{R}}^2] &= \sum_{R_1, R_2 \in \mc{R}}\E_{x}[\ind(x \in R_1) \cdot \ind(x \in R_2)] \\
    &=  \sum_{R \in \mc{R}}\E_{x}[\ind(x \in R)] +  \sum_{R_1 \neq R_2}\E_{x}[\ind(x \in R_1) \cdot \ind(x \in R_2)] \\
    &\leq \frac{M}{n} + \frac{M^2}{n^2}.
    \end{align*}
    It follows that,
    \[
    \Var(N_{\mc{R}}) = \left(\E_{x}[N_{\mc{R}}]\right)^2 - \E[N_{\mc{R}}^2] \leq \frac{M}{n}.
    \]
\end{proof}
Applying Chebyshev's inequality, we get the following.
\begin{lemma} \label{lm: chebyshev sampling}
    For any $c > 0$, it holds that
    \[
    \Pr_{x}\left[\left|N_{\mc{R}}(x) - \frac{M}{n}\right| \geq c\cdot \frac{M}{n}\right] \leq \frac{n}{c^2 \cdot M}.
    \]
\end{lemma}
\begin{proof}
    This lemma follows from a direct application of Chebyshev's inequality combined with the variance bound in \cref{lm: var bound}.
\end{proof}
\begin{lemma} \label{lm: sampling}
    For any set of points $S \subseteq [n]^d$ and any set $\mc{R}$ of one-variable restrictions of size $M$, we have
    \[
   \frac{1}{2}\left(\mu(S) - \frac{4n}{M}\right) \leq \nu_{\mc{R}}(S)  \leq 2\mu(S) + \frac{5n}{M}.
    \]
\end{lemma}
\begin{proof}
 For each integer $i$, let 
\[
m_i =  \left|\left\{x \in [n]^d \; | \; 2^i \frac{M}{n} \leq N_{\mc{R}}(x) < 2^{i+1}\frac{M}{n}\right\} \right|.
\]
By \cref{lm: var bound} we get that
\[
\frac{m_i}{n^d} \leq \frac{1}{(2^i-1)^2 M/n}.
\]

Towards the upper bound, we expand $\nu_{\mc{R}}(S)$ and perform a dyadic partitioning:
\begin{align*}
    \nu_{\mc{R}}(S) = \sum_{x \in S} \nu_{\mc{R}}(x) 
    = \sum_{x \in S} \frac{N_{\mc{R}}(x)}{Mn^{d-1}} 
    \leq  \frac{1}{Mn^{d-1}} \left(2\cdot \frac{M}{n} \cdot |S|+\sum_{i = 1}^{\infty} m_i 2^{i+1} \cdot \frac{M}{n} \right).
\end{align*} 
The first term in the parenthesis contributes 
$2\mu(\mathcal{S})$, whereas the contribution of
the second summand can be upper bounded by
\[
    \frac{1}{n^d}\cdot \sum_{i = 1}^{\infty} m_i 2^{i+1}
    \leq \frac{n}{M}\cdot \sum_{i=1}^{\infty} \frac{2^{i+1}}{(2^i-1)^2} 
    \leq \frac{5n}{M}
\]
For the lower bound, we have
\[
\nu_{\mc{R}}(S) = \sum_{x \in S}\nu_{\mc{R}}(x) \geq  \sum_{x \in S, N_{\mc{R}} \geq M/(2n)}\nu_{\mc{R}}(x) \geq |\{x \in S \; | \; N_{\mc{R}}(x) \geq n/(2M) \}| \cdot \frac{1}{2n^{d}}.
\]
To lower bound the last term above, we use \cref{lm: chebyshev sampling} to get
\[
\Pr_{x}[N_{\mc{R}}(x) < \frac{M}{2n}] \leq \frac{4n}{M}.
\]
It follows that 
\[
|\{x \in S \; | \; N_{\mc{R}}(x) \geq M/(2n) \}| \geq n^d\left((\mu(S) - \frac{4n}{M}\right).
\]
Putting everything together, we get the desired lower bound:
\[
\nu_{\mc{R}}(S) \geq \frac{1}{2}\left(\mu(S) - \frac{4n}{M}\right).
\]
\end{proof}

\subsubsection{Proof \cref{lm: restriction}}
Suppose $\mc{R} = \{R_1, \ldots, R_M\}$ is the set of one-variable restrictions on which $f$ is $\varepsilon$-close to a direct sum. For each $R_i$, let $f_i$ be this direct sum, so 
\[
\dist(f_i, f|_{R_i}) \leq \varepsilon,
\]
for each $1 \leq i \leq M$. By assumption $M \geq 100n^2/\eps$. Note that if $\eps$ is large --- say $\eps \geq \frac{1}{100}$ --- then the result is trivial, so for the remainder of the proof we suppose $\eps < \frac{1}{100}$.

As a useful fact, we show that the set of direct sum functions has minimum distance $\frac{1}{n}$.
\begin{lemma} \label{lm: distance}
    For any two distinct direct sums $f$ and $g$, we have
    \[
    \dist(f, g) \geq \frac{1}{n}.
    \]
\end{lemma}
\begin{proof}
    Write $f(x) = \sum_{i=1}^d L_i(x_i)$ and $g(x) = \sum_{i=1}^d L'_i(x_i)$. First suppose for each $i$, $L_i - L'_i$ is constant, then since $f(x) \neq g(x)$, we have $\delta(f,g) = 1$ and we are done.

    Suppose that there is some $i$ such that $L_i - L'_i$ is not the constant function and without loss of generality, let this $i$ be $1$. Then,
    \[
    \Pr_{x \in [n]^d}[f(x) \neq g(x)] = \E_{x_2,\ldots, x_d}\left[\Pr_{x_1}[L_1(x) - L'_1(x) \neq \sum_{i=2}^d L_i(x_i) - L'_i(x_i)]\right].
    \]
    To conclude, note that for any $x_2, \ldots, x_d$, the inner probability on the right hand side above is at least $1/n$ since $L_1 - L'_1$ is not the constant function.
\end{proof}

We now construct the following graph $G = (V, E)$. The vertex set is $V = \{1, \ldots, M\}$ while the set of edges $E$ consists of all pairs $(i,j)$ such that the functions $f_i$ and $f_j$ are consistent: 
\[
f_i|_{R_i \cap R_j} = f_j|_{R_i \cap R_j}.
\]
If $R_i \cap R_j$ is empty, then we also consider $f_i$ and $f_j$ consistent and have $(i,j) \in E$. We will first show that $G$ contains a large clique. Call a graph is \emph{transitive} if it is an edge disjoint union of cliques. Also define
\[
\beta(G) = \max_{(i,j) \notin E} \Pr_{k \in V}[(i,k), (j,k) \in E].
\]
Note that $G$ is transitive if and only if $\beta(G) = 0$. The following lemma from \cite{RS} gives a sort of approximate version of this fact and states that if $\beta(G)$ is small then $G$ is almost a transitive graph. 
\begin{lemma}
The graph $G$ can be made transitive by removing at most $3\sqrt{\beta(G)}|V|^2$ edges.
\end{lemma}

In our next two lemmas we show that $G$ has many edges, and then bound $\beta(G)$. This will establish that $G$ contains a large clique, which corresponds to many local direct sums $f_i$ that are consistent with one another.

\begin{lemma} \label{lm: many edges}
The graph $G$ has at least $0.9M^2$ edges.
\end{lemma}
\begin{proof}
    For each $R_i$, let $S_i \subseteq R_i$ be the set of points on which $f_i$ and $f$ differ. We have that
    \[
    \frac{|S_i|}{|R_i|} \leq \varepsilon,
    \]
    by assumption.
    
    Now fix an $R_i$. We will apply \cref{lm: sampling} inside of $R_i$, which is isomorphic to the grid $[n]^{d-1}$. For each $j$, let $R'_j = R_i \cap R_j$. Let $\mc{R'} = \{R'_j \; | \; \frac{R'_j \cap S_i}{R'_j} \geq 3 \varepsilon \}$, and let $M' = |\mc{R'}|$. Now consider the size of $S_i$ inside of $R_i$ under the measure $\nu_{\mc{R'}}$. We have,
    \[
    \nu_{\mc{R'}}(S_i) \geq 3 \varepsilon.
    \]
     By \cref{lm: sampling}
    \[
    3\eps \leq \nu_{R'}(S_i) \leq 2\eps + 5n/M'.
    \]
    It follows that,
    \[
    M' \leq \frac{5n}{\eps}.
    \]
     Thus, 
    \[
    \Pr_{R_i, R_j}\left[\frac{|R_i \cap R_j \cap S_i|}{|R_i \cap R_j|} \geq 3 \varepsilon\right] \leq \frac{5n}{\eps M}.
    \]
    By a union bound, it follows that with probability at least $1 - \frac{10n}{\eps M}$, we have that both
    \[
    \frac{|R_j \cap R_i \cap S_i|}{|R_j \cap R_i|} \leq 3 \varepsilon \quad \text{and} \quad   \frac{|R_j \cap R_i \cap S_j|}{|R_j \cap R_i|} \leq 3 \varepsilon.
    \]
    For such $i, j$, we get that 
    \[
    \Pr_{x \in R_i \cap R_j}[f_i(x) = f_j(x)] \geq 1 - \Pr_{x \in R_i \cap R_j}[x \in S_i \cup S_j] \geq 1 - 6 \varepsilon > \frac{1}{n},
    \]
     and $f_{i}|_{R_i \cap R_j} = f_j|_{R_i \cap R_j}$ by \cref{lm: distance}. Thus choosing $i,j$ randomly, we get that they are adjacent in $G$ with probability at least $1 - \frac{10n}{\eps M} \geq 0.9$.
\end{proof}
\begin{lemma} \label{lm: bound beta}
    Suppose $R_i, R_j$ are distinct one-variable restrictions such that $f_{i}|_{R_i \cap R_j} \neq f_{j}|_{R_i \cap R_j}$. Then,
    \[
    \Pr_{R_{k}}[f_{i}|_{R_i \cap R_j \cap R_k} = f_{j}|_{R_i \cap R_j \cap R_k}] \leq \frac{5n^2}{M} \leq \frac{\eps}{20}.
    \]
\end{lemma}
\begin{proof}
    Let $W = R_i \cap R_j$. For each $k$ let $R'_k = R_k \cap W$, let $f'_i = f_i|_W$, and let $f'_j = f_j|_W$.
    
    Throughout the proof we will consider $W$ as our ambient space and we will apply \cref{lm: sampling} inside of $W$. Note that we may do so because $W$ is isomorphic to the grid $[n]^{d-2}$. Define the following set of one-variable restrictions inside $W$:
    \[
    \mathcal{R}' = \{W \cap R_k \neq \emptyset \; | \; R_k \in \mc{R}, f'_i|_{R'_k} = f'_j|_{R'_k} \}.
    \]
    We will show that $|\mc{R}'|$ cannot be too large. Indeed, let $M'$ denote its size and let $S \subseteq W$ be
    \[
    S = \{x \in W \; | \; f'_i(x) \neq f'_j(x) \}.
    \]
    We have that $\frac{|S|}{|W|} \geq \frac{1}{n}$. On the other hand, by definition of $\mc{R}'$, we have that $\nu_{R'}(S) = 0$. Therefore by \cref{lm: sampling}:
    \[
    0 \geq \frac{1}{2}\left(\frac{1}{n} - \frac{4n}{M'} \right).
    \]
    It follows that 
    \[
    M' \leq 4n^2.
    \]
    Finally there can be up to $2n$ additional restrictions $R_k$ such that $R_k \cap R_i = \emptyset$ or $R_k \cap R_j = \emptyset$. The result follows.
\end{proof}

Combining \cref{lm: many edges} and \cref{lm: bound beta}, we get that $G$ contains a transitive subgraph with at least $0.8M^2$ edges. Let $\mc{C}_1, \ldots, \mc{C}_J$ be the edge disjoint union cliques of this transitive subgraph and say $\mc{C}_1$ is the largest one. It follows that,
\[
0.8M^2 \leq \sum_{I=1}^J |\mc{C}_i|^2 \leq |\mc{C}_1| \cdot M,
\]
and $G$ contains a clique of size $M_1 = 0.8M$. From now we let $\mc{C}$ denote this clique and let $\mc{C}= \{1, \ldots, M_1\}$. We define the following function $F$ using the functions $f_1, \ldots, f_{M_1}$: set $F(x) = f_j(x)$ if there is some $j \in \mc{C}_1$ such that $x \in R_j$, and $F(x) = 0$ otherwise.

To conclude, we will show that $F$ is close to $f$ and $F$ is close to a direct sum. This will establish that $f$ is close to a direct sum.

\begin{lemma} \label{lm: F close to f}
$F$ is $3\eps$-close to $f$.
\end{lemma}
\begin{proof}
Let $S = \{x \in [n]^d \; | \; f(x) \neq F(x) \}$. Since $f$ is $\eps$-close to $F$ on the restrictions $R_i$ for all $i \in \mc{C}$, we have
\[
\nu_{\mc{C}}(S) \leq \eps.
\]
By \cref{lm: sampling}, it follows that, 
\[
\eps \geq \frac{1}{2}\left(\mu(S) -  \frac{4n}{0.8M}\right),
\]
and 
\[
\mu(S) \leq 3\eps.
\]
\end{proof}

\begin{lemma} \label{lm: close to ds}
$F$ is $O(\eps)$-close to a direct sum.
\end{lemma}
\begin{proof}
We will show that $F$ passes the square in a cube test with high probability. Recall this test operates by choosing $a,b \in [n]^d$, $x,y \in \{0,1\}^d$ and checking if $F$ satisfies,
\[
F(a) + F(\phi_x(a,b)) + F(\phi_{y}(a,b)) + F(\phi_{x+y}(a,b)) = 0.
\]
Note that $F$ passes if there is some $i \in \mc{C}$ such that $a,b \in R_i$. Indeed, in this case all of the queried points are in $R_i$, so $F$ and the direct sum $f_i$ agree on all of the queried points and
\[
F(a) + F(\phi_x(a,b)) + F(\phi_{y}(a,b)) + F(\phi_{x+y}(a,b)) = f_i(a) + f_i(\phi_x(a,b)) + f_i(\phi_{y}(a,b)) + f_i(\phi_{x+y}(a,b)) = 0.
\]
To bound the probability that $i \in \mc{C}$ such that $a,b \in R_i$, let $N'(x) =| \{i \in \mc{C} \; | \; x\in R_i \}|$. By \cref{lm: chebyshev sampling},
\[
\Pr_{a}\left[N'(a) \geq \frac{M}{1.6n}\right] \leq \frac{5n}{M}.
\]
Now condition on $a$ satisfying $N'(a) \geq \frac{M}{1.6n}$ and let $\mc{C}'' = \{i \in \mc{C} \; | \; a \in R_i \}$. We label $\mc{C}'' = \{1,\ldots, M''\}$ where $M'' \geq \frac{M}{1.6n}$ and let $N''(x) = |\{i \in \mc{C}'' \; | \; x \in R_i \}|$. Then by \cref{lm: chebyshev sampling} again,
\[
\Pr_{b}[N''(b) \geq 0] \leq \frac{n}{M''} \leq \frac{1.6n^2}{M}.
\]

By a union bound, with probability at least $1-\frac{7n^2}{M}$ we have that $a$ and $b$ are both contained in some $R_i$ with $i \in \mc{C}$. It follows that $F$ passes the square in a cube test with probability at least $1-\frac{7n^2}{M}$. By \cref{thm: sq in cu}, $F$ is $O(\frac{n^2}{M}) = O(\eps)$-close to a direct sum. 
\end{proof}
\begin{proof}[Proof of \cref{lm: restriction}]
    The result follows by combining \cref{lm: F close to f}, \cref{lm: close to ds} and using the triangle inequality.
\end{proof}

%% file: fourierfact.tex
\subsection{The Fourier Analytic Interpretation of the Diamond Test} \label{sec: fourier}

In this section, we give a Fourier analytic interpretation of the Diamond test.
Let $\mathsf{Even}$ denote the class of functions $f:\F_2^d\to\F_2$ satisfying
$f(x+\one)=f(x)$ and let $\mathsf{Odd}$ denote the class of functions
$f:\F_2^d\to\F_2$ satisfying $f(x+\one)=f(x)+1$. 
Let $\eo = \mathsf{Even}\cup\mathsf{Odd}$.
We note that $\eo$ has a nice Fourier analytic interpretation: 
$f\in \eo$ implies that $(-1)^f$ has either all its Fourier-mass on even-sized
characters, or all of its Fourier-mass on odd-sized characters.
Given a function $f:\F_2^d\to\F_2$, define random variable $\mathcal{R}(f)$ to
be a random restriction of $f$ obtained as follows: for each $i\in [d]$
independently, with probability $1/4$ restrict $x_i$ to $0$, with probability
$1/4$ restrict $x_i$ to $1$, and with probability $1/2$ do not restrict $x_i$.
We now show that the Diamond test's correctness is equivalent to the following Fourier analytic fact

\begin{theorem} \label{conj:fourier}
$$\dist(f,\mathsf{affine}) \le O(\E[\dist(\mathcal{R}(f), \eo)])  $$
\end{theorem}
% \begin{theorem} \label{conj:fourier}
%     \[
%     \max_{S\subseteq [d]} |\widehat{f}(S)| \leq C \cdot \E_{R}\left[\left(\sum_{|S| \text{ even}}\widehat{f|_R}(S)^2\right)^2 + \left(1- \sum_{|S| \text{ even}}\widehat{f|_R}(S)^2\right)^2\right]
%     \]
% \end{theorem}
In fact, we show something stronger: there is a very tight quantitative relationship between $$\E[\dist(\mathcal{R}(f), \eo)]$$ and the probability that the Diamond test rejects $f$. 

\begin{lemma}\label{lem:frr}
Fix $f$. Let $\eps$ be the probability that the Diamond test rejects $f$, and let $$\delta = \E[\dist(\mathcal{R}(f),\eo)].$$ Then, 
$$2\delta  \le \eps \le 4\delta.$$
\end{lemma}
\begin{proof}
Fix $f,\eps,\delta$ as in the lemma statement.
Given $a,b$ define $C_{a,b}=\setof{\phi_x(a,b)}{x\in\F_2^d}$ and $f_{a,b}(x) = f(\phi_x(a,b))$.
We have:
\begin{equation}\label{eq:restrict}
\eps = \Pr_{a,b,x}[f_{a,b}(x)+f_{a,b}(x+\one) \neq f_{a,b}(0)+f_{a,b}(\one)].
\end{equation}
\begin{claim}
$\eps \le 4\delta$.
\end{claim}
Let $\delta_{a,b}=\dist(f_{a,b}, \eo)$ .  
By a union bound we have
$$\Pr_{z,x}[f_{a,b}(x+\one)+f_{a,b}(x)\neq f_{a,b}(z+\one)+f_{a,b}(z)]\le 4\delta_{a,b}.$$
By averaging, $\delta = \E_{a,b}[\delta_{a,b}]$. Thus, 
\begin{equation}
   \label{eq:frrrr} 
\Pr_{a,b,z,x}[f_{a,b}(x+\one)+f_{a,b}(x)\neq f_{a,b}(z+\one)+f_{a,b}(z)]\le 4\delta.
\end{equation}
Now, observe that for random $a,b,z,x$ the following two distributions are the same:
\begin{enumerate}
    \item $(\phi_x(a,b),\phi_x(b,a),a,b)$
    \item $(\phi_x(a,b),\phi_x(b,a),\phi_z(a,b),\phi_{z}(b,a))$.
\end{enumerate}
Thus \cref{eq:frrrr} implies $\eps \le 4\delta$.

\begin{claim}
$\eps\ge 2\delta$.
\end{claim}
For any function $g$, 
$$\Pr_x[g(x)\neq g(x+\one)]=2\dist(g,\mathsf{Even}),$$
$$\Pr_x[g(x)= g(x+\one)]=2\dist(g,\mathsf{Odd}).$$
Thus, for any $a,b$ we have
$$\Pr_x[f_{a,b}(x)+f_{a,b}(x+\one)\neq f_{a,b}(0)+f_{a,b}(\one)]\ge 2\dist(f_{a,b},\eo).$$
Averaging over $a,b$ on both sides gives $\eps \ge 2\delta$.

% Fix some $a,b$ with $f_{a,b}(0)+f_{a,b}(1) = 0$, and let
% $$
% \eps_{a,b} = \Pr_x[f_{a,b}(x)+f_{a,b}(x+\one) \neq 0].
% $$
% Switching to multiplicative notation for convenience we have
% $$\E_x [(-1)^{f_{a,b}(x)}(-1)^{f_{a,b}(-x)}] = 1-2\eps_{a,b}.$$
% Applying Plancherel's Identity ($\ang{f,g}=\sum_S \hat{f}(S)\hat{g}(S)$, see \cite{od21}) we find:
% $$\sum_{|S| \text{ even}} \widehat{f_{a,b}}(S)^2 - \sum_{|S|\text{ odd}}
% \widehat{f_{a,b}}(S)^2 = 1-2\eps_{a,b}.$$
% Because $\sum_{S}\hat{f}(S)^2=1$ we have
% $$\sum_{|S|\text{ even}} \widehat{f}(S)^2 = 1-\eps_{a,b}.$$
% That is, $f$ is $\eps_{a,b}$-close to being even.

% The case of $a,b$ with $f_{a,b}(0)+f_{a,b}(1)=1$ is essentially the same. Define 
% $$
% \eps_{a,b}' = \Pr_x[f_{a,b}(x)+f_{a,b}(x+1) \neq 1].
% $$
% Then, 
% $$
% \E_x[(-1)^{f_{a,b}(x)}(-1)^{f_{a,b}(-x)}]  = 2\eps_{a,b}'-1.
% $$
% Applying Plancherel's Identity again, we find that $f_{a,b}$ is
% $\eps_{a,b}'$-close to being odd.

% By averaging, $\E[\eps_{a,b}]=\eps$.
% Thus, the average distance of $f_{a,b}$ to $\eo$ over random $a,b$ is equal to $\eps$.
% \todo{here.}
% In summary, the expected distance of $f_{a,b}$ to $\eo$ over random $a,b$ is at most $\eps$.
% Assuming \cref{conj:fourier} this implies that $f$ is $O(\eps)$-close to linear.

\end{proof}

\begin{proof}[Proof of \cref{conj:fourier}]
The result follows from the soundness of the Diamond Test, \cref{3thm:diamond_cube_soundness}
\end{proof}

%% file: subcube.tex
\section{The Affine on Subcube Test}\label{sec:subcube}
% In \cite{di19} Dinur and Golubev define the \defn{Square in Cube} test.
% \begin{test}[Square in Cube]\label{def:squarecube}
% Sample random $a,b\in [n]^d$ and $x,y\in \mathbb{F}_2^{d}$. Accept iff:
% $$  f(\phi_0(a,b)) + f(\phi_x(a,b))+ f(\phi_y(a,b))=f(\phi_{x+y}(a,b)). $$ 
% \end{test}

In this section we consider a class of direct sum tests that generalizes
the Square in Cube test of \cite{di19}. We call these tests \defn{Affine on
Subcube} tests. Indeed, the Square in Cube test can be viewed as performing the
BLR affinity test inside of a random induced hypercube inside of $[n]^d$. We
show that in fact testing for affinity on a random subcube with \emph{any}
affinity tester also works as a direct sum test. 
\begin{test}[Affine on Subcube]\label{test:affsubcube}
Sample $a,b\in [n]^d$ and run an affinity test that has soundness function $s(\delta) = \Omega(\delta)$ on the function $x\mapsto f(\phi_x(a,b))$.
\end{test}

The flexibility of being able to use any affinity tester also enables us to
obtain a direct sum test with additional properties such as erasure resiliency,
which we discuss further in \cref{sec:erasure}, and less use of randomness, which we discuss in \cref{sec:derandomize}. Furthermore, we believe that our
analysis for going from local direct sums to a global direct sum is simpler than
that of \cite{di19} and thus noteworthy on its own.

One final benefit of our analysis is that it fixes an error in \cite{di19}'s
analysis. Specifically, they establish that a function $f$ passing the Square in
Cube test with probability $1-\eps$ is \emph{locally} close to a direct sum
(this is basically \cref{lem:fabdirectprod1} in our presentation), but they do
not show how to combine these local direct sums into a global direct sum.

We now show that the Affine on Subcube test is a valid direct sum tester. Once
again, completeness is clear. 
If $f$ is a direct sum, then the Affine on Subcube test
passes with probability $1$. 
Our focus is then on showing soundness.
\begin{theorem}\label{thm:subcube}
Fix $f:[n]^d\to \F_2, \eps\ge 0$ and an affinity test $T$. If $f$
passes the Affine on Subcube test (instantiated with $T$) with
probability $1-\eps$, then $f$ is $C\eps$-close to a direct sum, for some
constant $C$ independent of $n, d$.
\end{theorem}
\begin{proof}
Fix $f$ as in the theorem statement. That is, satisfying
\begin{equation}\label{eq:step1}
\E_{a,b\sim [n]^d}[\dist(x\mapsto f(\phi_x(a,b)), \mathsf{affine})] \le O(\eps).
\end{equation}
We re-interpret \cref{eq:step1} in three steps.
First, observe that \cref{eq:step1} is equivalent to 
\begin{equation}\label{eq:step2}
\E_{a,b\sim [n]^d,z\sim\F_2^d}[\dist(x\mapsto f(\phi_z(a,b)) + f(\phi_{x+z}(a,b)), \mathsf{linear})] \le O(\eps).
\end{equation}
Next, observe that we can rewrite $\phi_{x+z}(a,b)$ as follows:
$$\phi_{x+z}(a,b) = \phi_x(\phi_z(a,b), \phi_z(b,a)).$$
In light of this observation, \cref{eq:step2} is equivalent to
\begin{equation}\label{eq:step3}
\E_{a,b\sim [n]^d,z\sim\F_2^d}[\dist(x\mapsto f(\phi_z(a,b)) + f(\phi_x(\phi_z(a,b), \phi_z(b,a))), \mathsf{linear})] \le O(\eps).
\end{equation}
Of course, the distribution of $(\phi_z(a,b),\phi_z(b,a))$ is the same as the distribution of $(a,b)$.
So, we can rewrite \cref{eq:step3} as:
\begin{equation}\label{eq:step4}
\E_{a,b\sim [n]^d}[\dist(x\mapsto f(a) + f(\phi_x(a, b)), \mathsf{linear})] \le O(\eps).
\end{equation}
This interpretation of the test is more convenient to work with.

For each $a,b$ define $f_{ab}$ to be the function $x\mapsto
f(a)+f(\phi_x(a,b))$, and let $\eps_{ab}$ be the probability that $f_{ab}$ fails
the linearity test. Clearly $\E[\eps_{ab}] = \E[\eps_a] \le O(\eps).$
Because $f_{ab}$ passes the linearity test with probability $1-\eps_{ab}$, there
is vector $F^a(b) \in \F_2^d$ such that 
\begin{equation}\label{eq:restrictionbenice}
\Pr_x[f_{ab}(x)=\langle F^a(b), x\rangle]>1-O(\eps_{ab}),
\end{equation}
where $\ang{v,w} = \sum_i v_i w_i.$ 
We will now show that $b\mapsto F^a(b)$ is close to a direct product, by proving
that this function passes a certain direct product test. 
Dinur and Golubev \cite{di19} give the following direct product test (which is slightly more convenient here than the original direct product test of \cite{dinur2014direct}):
\begin{test}[Direct Product Test]\label{test:dp}
Sample $x\in [n]^d$. Sample a set $A$ by including each $i\in [d]$ in $A$ with probability $3/4$.
Sample $y$ by setting $y_i=x_i$ if $i\in A$, and otherwise sampling $y_i$ uniformly from $[n_i]$.
Accept iff $f(x)_i = f(y)_i$ for all $i\in A$.
\end{test}
Using \cite{dikstein2019agreement}, Dinur and Golubev show
\begin{fact}\label{fact:directproduct}
If $f$ passes \cref{test:dp} with probability $1-\eps$, then $f$ is
$O(\eps)$-close to a direct product.
\end{fact}

Now, we use \cref{fact:directproduct} to show that $b\mapsto F^a(b)$ is close to a direct product.
\begin{lemma}\label{lem:fabdirectprod1}
$b\mapsto F^a(b)$ is $O(\eps_a)$-close to a direct product.
\end{lemma}
\begin{proof}
Given $b,b'$ define $\mc{D}_{bb'}$ on $\mathbb{F}_2^d$ as follows. If $x\sim \mc{D}_{bb'}$
then for each $i\in [d]$ independently, $x_i$ is determined as follows:
\begin{quote}
If $b_i \neq b_i'$,  then $x_i = 0$.
If $b_i = b_i'$, then $x_i=1$ with probability $2/3$ and $0$ otherwise.
\end{quote}
\newcommand{\eqtf}{\mathsf{eq}_{3/4}}
Define $\eqtf = \T_{(3n/4-1)/(n-1)}$. If $b'\sim \eqtf(b)$, then for each $i$
independently, $b_i = b_i'$ with probability $3/4$.
Now, observe that if we sample $b'\sim \eqtf(b)$ and $x\sim \mc{D}_{bb'}$, then
the values $x_i$ are independent uniformly random bits.
However, if we choose $b',x$ in this manner, then the joint distribution of
$(b,b',x)$ satisfies
$$\phi_x(a,b)  = \phi_x(a,b').$$
This is because we select $a_i$ whenever $b_i\neq b_i'$.
Then, performing a union bound on \cref{eq:restrictionbenice} (and averaging
away the dependence on $b$) we have that 
\begin{equation}\label{eq:litunionbound}
\Pr_{b\sim [n]^d,b'\sim \eqtf(b),x\sim \mc{D}_{bb'}} [\ang{F^a(b),x} =f_{ab}(x) = f_{ab'}(x) = \ang{F^a(b'),x}] > 1-O(\eps_{a}).
\end{equation}
For each $a,b,b'$, define
$$
\delta_{abb'} =  \Pr_{x\sim \mc{D}_{bb'}} [\ang{F^a(b)+F^a(b'), x} \neq 0].
$$
By \cref{eq:litunionbound}, $\E_{b\sim [n]^d,b'\sim \eqtf(b)}[\delta_{abb'}]< O(\eps_a)$.
Intuitively, if $\delta_{abb'}$ is small then $F^a(b),F^a(b')$ must have some ``agreement''. More precisely we have:
\begin{claim}
If $\delta_{abb'}<1/3$ , then for all $i$ where $b_i = b_i'$ we have $F^a(b)_i = F^a(b')_i$
\end{claim}
\begin{proof}
Suppose to the contrary that there is some $i\in [d]$ with $b_i= b_i'$ but $F^a(b)_i \neq F^a(b')_i$.
If $x\sim \mc{D}_{bb'}$ then there is a $1/3$ chance of $x_i = 0$ and a $2/3$ chance
of $x_i=1$, and this is independent of the values of $x_j$ for $j\neq i$. This means that $\ang{F^a(b)+F^a(b'), x}$ takes on either value in $\F_2$ with probability at least $1/3$, implying that $\delta_{abb'}\ge 1/3$.
\end{proof}
Now, by Markov's Inequality we have 
\begin{equation}\label{eq:markovnice}
\Pr_{b\sim[n]^d,b'\sim \eqtf(b)}[\delta_{abb'}\ge 1/3] \le 3\mathbb{E}_{b,b'}[\delta_{abb'}] < O(\eps_a).
\end{equation}
Given $b,b'$ let $\overline{\Delta}(b,b')$ denote the set of coordinates $i$ where $b_i=b_i'$. By \cref{eq:markovnice} we have
$$
\Pr_{b\sim[n]^d,b'\sim \eqtf(b)} [F^a(b)_i = F^a(b')_i \quad \forall i\in \overline{\Delta}(b,b')] > 1-O(\eps_a).
$$
By  \cref{fact:directproduct} we conclude that $F^a(b)$ is $O(\eps_a)$-close to a direct product. 
\end{proof}

Now we conclude the proof of \cref{thm:subcube}.
For each $a$, let $(g_1(a,b_i),\ldots, g_d(a,b_d))$ be a direct product that is
close $O(\eps_a)$-close to $F^a(b)$; the existence of these direct products was
shown in \cref{lem:fabdirectprod1}. Then, 
\begin{equation}\label{eq:giabdpdp}
\Pr_{a,x,b} \left[ f_{ab}(x) = \sum_i g_i(a,b_i)x_i \right] > 1-O(\eps).
\end{equation} 
Observe that $\phi_{x}(a,b)=\phi_{\one + x}(b,a)$. Thus, by \cref{eq:giabdpdp}
and a union bound we have
\begin{equation}\label{eq:flipflip}
\Pr_{a,b,x} \left[ \sum_{i}g_{i}(a,b_{i})x_{i}+ f(a)=f_{ab}(x)=f_{ba}(x+\one)
=f(b)+\sum_{i}g_{i}(b,a_{i})(1+x_{i}) \right]>1-O(\eps).
\end{equation}
For each $a,b$, define 
$$
\lambda_{ab}=\Pr_x\left[ \sum_i (g_i(a,b_i)+g_i(b,a_i))x_i \neq \sum_i g_i(b,a_i) + f(a)+f(b) \right].
$$
By \cref{eq:flipflip}, $\E_{ab}[\lambda_{ab}]\le O(\eps)$.
If $\lambda_{ab}<1/2$  then we must have 
$$
\sum_i g_i(b,a_i)=f(a)+f(b).
$$
By Markov's inequality:
$$
\Pr_{a,b}[\lambda_{ab}\ge 1/2]\le 2\mathbb{E}_{a,b}[\lambda_{ab}]\le O(\eps).
$$
Thus, we have:
$$
\Pr_{a,b}\left[ f(a)=f(b)+\sum_i g_i(b,a_i)  \right] > 1-O(\eps).
$$
Then, by the probabilistic method there exists $\beta$ such that 
$$
\Pr_a\left[ f(a)=f(\beta)+\sum_i g_i(\beta,a_i) \right]>1-O(\eps).
$$
That is, $f$ is $O(\eps)$-close to the direct sum $a\mapsto f(\beta)+\sum_i g_i(\beta,a_i)$.
\end{proof}  
\input{derandomize}
\input{online_adversary}

%% file: derandomize.tex
\subsection{Derandomization} \label{sec:derandomize}
In this section we show how to use \cref{thm:subcube} to obtain a more randomness-efficient direct sum tester than the Square in Cube test or Diamond test. 
In particular, we define the \defn{Diamond in Cube test} as follows: run the
Diamond test on the restriction of $f$ to a random subcube (i.e.,
$f(\phi_x(a,b))$). 

% In this section we conclude the proof of soundness in \cref{thm:diamond_cube_soundness} by combining the results of the previous two sections. We will view the overall Diamond in the Cube test as performing the Diamond Test inside of a random induced hypercube inside of $[n]^d$. Note however, that the Diamond in Cube test is \emph{not} equivalent to performing the standard, unbiased
% Diamond test on a random subcube. Instead, for functions $f:[n]^{d}\to \F_2$,
% performing a Diamond in Cube test on $f$ is equivalent to performing a
% $(2/n-1)$-Biased Diamond test on a random subcube.
% Combining \cref{cor:1fun} and \cref{thm:subcube} we have:
% \begin{corollary}\label{cor:dcube}
% If $f:[n]^{d}\to \F_2$ passes the Diamond in Cube test with probability $\eps$,
% then $f$ is $O_n(\eps)$-close to a direct sum.
% \end{corollary}
% \cref{cor:dcube} could also be generalized to the case of $f:[n;d]\to \F_2$, by
% analyzing a $p$-Biased Diamond test with a different bias at each coordiante;
% however we choose not to do this.
% \cref{cor:dcube} resolves the open question of Dinur and Golubev.
% However, we can also use the $0$-Biased Diamond test to get a direct sum test,
% by simply running the $0$-Biased Diamond test on $f(\phi_x(a,b))$ for random
% $a,b\in [n;d]$.
% We will call this the \defn{Direct-Sum Diamond test}. 
We now show that the Direct-Sum Diamond test is more randomness-efficient than
the Square in Cube test and the Diamond in Cube test for some parameter regimes.
In the following discussion we assume that $n$ is a power of two for convenience.

In the statement of tests like the Square in Cube test we say ``sample $a,b\sim
[n]^d,x,y\sim \F_2$''. This would seem to indicate that the test requires $(2+2\log_2 n)d$
bits of randomness. However, if we look inside the Square in Cube test, we see
that this is not the case. The Square in Cube test actually requires sampling
from the following distribution:
$$(\phi_0(a,b),\phi_x(a,b),\phi_y(a,b),\phi_{x+y}(a,b)),$$
where $a,b\sim [n]^d$ and $x\sim \F_2^d$.
However, sampling from this distribution does not actually require sampling
$a,b,x,y$ fully.
Instead, we can sample from the distribution as follows.
First, sample $a\sim[n]^d, x,y\sim\F_2^d$. 
Then we partially sample $b$. Note that we don't need to sample $b_i$ if
$x_i=y_i=0$, because in this case the value of $b_i$ will not be needed in any
of the query points. Thus, the average number of random bits required to create
a sample for the Square in Cube test is only $(2+1.75\log_2 n)d$.

On the surface the Diamond test might appear to be more efficient.
Rather than sampling $a,b,x,y$ it only needs to sample $a,b,x$.
However, the Diamond test actually requires more random bits than
the Square in Cube test if $n\ge 20$.
The number of random bits required to sample from the Diamond in Cube's
distribution is $(2\log_2 n + (1-1/n))d$.

Now, we claim that the Diamond in Cube test is more randomness-efficient than
both the Diamond test and the Square in Cube test, at least for large $n$. A
simple computation shows that we can obtain samples from the Diamond in Cube
test's sample distribution using only $(2.5+1.5\log_2 n)d$ random bits,
vindicating this claim.

%% file: online_adversary.tex
\subsection{Direct Sum Testing in the Online Adversary Model} \label{sec:erasure}

In this section we show that an immediate application of the Affine on Subcube test from \cref{thm:subcube} yields direct sum testers that are online erasure-resilient. The $t$-online erasure model was first considered by Kalemaj, Raskhodnikova, and Varma in \cite{KRV} and was further studied in \cite{MZ, BKMR}. The properties that these works consider include linearity, low degree, and various properties of sequences. Let us now describe the model.

We are once again given query access to an input the evaluation of a function $f: [n]^d
\to \F_2$, and the goal is the same as in direct sum testing --- to
distinguish whether $f$ is a direct sum, or far from a direct sum. The twist is
that in the $t$-online erasure model, there is an adversary who is allowed to
erase any $t$ entries of $f$ after each query that the property testing
algorithm makes. If a point is queried after it is erased, then $\perp$ is
returned instead of the actual value, i.e\ the oracle will return $\perp$ instead of the actual value $f(x)$ if that query is made after the oracle erases the $x$ entry of $f$'s evaluation. It is not hard to see that neither the
Diamond in the Cube test we analyze nor the Square in the Cube test of
\cite{di19} is erasure resilient. The reason is that in both tests, after the first three queries are made, the adversary knows what the fourth query will be and can erase $f$ at that point.

However, by using the Affine on Subcube test
with an erasure resilient affinity tester, we obtain an erasure resilient direct
sum tester.

\begin{theorem}
  Fix a distance parameter $\delta > 0$ and erasure parameter $t$, and take $d$ sufficiently large relative to $n, \delta, t$. Then there is a direct sum tester for $f: [n]^d \to \F_2$ that makes $O(\max(1/\eps, \log t))$-queries in the $t$-online erasure model that satisfies:
\vspace{0.2cm}
  
  \begin{itemize}
      \item \textbf{Completeness:} If $f$ is a direct sum, the tester accepts with probability $1$.
      \item \textbf{Soundness:} If $f$ is $\delta$-far from being a direct sum, then the tester rejects with probability $2/3$.
  \end{itemize}
\end{theorem}
\begin{proof}
    The result follows by combining \cref{thm:subcube} with Theorem 1.1
    of \cite{BKMR}. Specifically we use the Affine on Subcube test where the
    inner affinity tester is the erasure resilient one of \cite{BKMR}.
    Technically, Theorem 1.1 of \cite{BKMR} is stated for linearity testing, but
    it is straightforward to modify their result to obtain an affinity tester.
\end{proof}

%% file: self-correction.tex
\section{Reconstruction for Direct Sums}\label{sec:correction}
In this section we address another question raised in \cite{di19} on whether a test, which they call the \defn{Shapka Test} can be used to reconstruct the underlying direct sum. We show that
the this is indeed the case in some parameter regimes (i.e\ with the input function being sufficiently close to a direct sum), and then we
give an alternate method for obtaining corrected versions of the direct sum,
that is both more query efficient and more error tolerant than the Shapka test. We also give upper bounds that match up to a constant for the fraction of errors that can be tolerated when reconstructing the direct sum. This shows that the new reconstruction method we propose is essentially optimal: 
it works for essentially the entire range of $\eps$ where recovery is
information theoretically possible, and there is no
asymptotically more query-efficient correction method.

We start with our lower bounds.
\begin{proposition}\label{prop:infolb}
Fix an even $n$, arbitrary integer $d$, and some $\epsilon > 0 $ such that $n\eps>1/2$. There exist distinct direct sums
$L,Z:[n]^d\to \F_2$, and a function $f:[n]^d\to \F_2$
such that $\dist(f,L) = \dist(f, Z) \le \eps$.
\end{proposition}
\begin{proof}
Let $Z$ be the zero function, and let $L$ be the indicator of $x_d = 1$.
Let $f$ be any function which is zero if $x_d\neq 1$, and for inputs with
$x_d=1$, $f$ is zero half the time.
Then, $\dist(f, L) = \dist(f, Z) = 1/(2n) \le \eps$.
In other words, there is no unique closest direct sum to $f$, so decoding $f(\one)$ is undefined.
\end{proof}

\begin{proposition}\label{prop:querylb}
Fix an $n$ divisible by $4$, an arbitrary integer $d$, and $\epsilon > 0$ such that $(0.7)^{d/n}< \eps, n\eps < 1/2$. Then there is a set
$\mathcal{F}$ of functions $f: [n]^d\to \F_2$, such that
using fewer than $n/4$ queries it is impossible to
determine the value of $L(\one)$ with probability greater than $1/2$, where $L$
is the closest direct sum to $f$. 
\end{proposition}
\begin{proof}
Call a point $x\in [n]^d$ \defn{heavy} if more than $2d/n$ coordinates $i\in
[d]$ have $x_i=1$. By a Chernoff bound, the fraction of points in $[n]^d$ which
are heavy is less than $(0.7)^{d/n}$. Let $L$ be a uniformly random direct sum.
Define $\mathcal{F}$ to be the class of all functions which agree with $L$ on
non-heavy inputs. Then, querying $f\in \mathcal{F}$ on a heavy point is useless,
because it outputs an arbitrary value unrelated to $L$. By our assumption that
$(0.7)^{d/n}<\eps$  we have that any $f\in \mathcal{F}$ is $\eps$-close to $L$.
Now we argue that $n/4$ non-heavy queries do not suffice to determine $L(\one)$
with non-trivial probability. Indeed, there will be a coordinate $i\in [d]$ such
that among the $n/4$ non-heavy queries, no queried point $x$ has $x_i = 1$.
Thus, we cannot hope to distinguish between the equally likely cases $L_i(1)=0$
and $L_i(1)=1$, and thus cannot predict $L(\one)$ with non-trivial probability.
\end{proof}

Now we give algorithms for reconstructing values from the closest direct sum to
a function $f$. Let us start with the \defn{Shapka Test} of \cite{di19} which we now describe. This test is slightly different for odd and even $d$; for simplicity we only consider the case of odd $d$.
Let $e_i\in \{0,1\}^d$ denote a vector with zeros except at the $i$-th coordinate. We describe how the Shapka Test can be used to reconstruct a direct sum. 

\begin{test}[Shapka \cite{di19}]\label{test:shapka}
Sample $a,b\in [n]^d$. Accept iff
\begin{equation}\label{eq:shapka}
f(b) = \sum_{i=1}^d f(\phi_{e_i}(a, b)).
\end{equation}
\end{test}

Thus, to reconstruct a direct sum from a corrupted version, $f$, we can define,
\[
f_{{\sf shapka}}(b) = \maj_{a \in [n]^d} \sum_{i=1}^{d}  f(\phi_{e_i}(a, b)).
\]
Dinur and Golubev show that if $f$ passes the Shapka test with probability
$\eps$, then $f$ is $O(n\eps)$-close to a direct sum.
Now, instead of checking \cref{eq:shapka}, we use the expression on the
right hand side of \cref{eq:shapka} as our guess for the value of $f(b)$, and take the most common value at each $b$ to define the reconstructed function $f_{{\sf shapka}}$. We show that if $f$ is close enough to a direct sum, then $f_{{\sf shapka}}$ is in fact this direct sum:
\begin{proposition}
Fix $\eps>0, n\in \N$, and $d$ odd with $n\eps d < 1/4$. Suppose $f:[n]^d\to \F_2$
 is $\eps$-close to the direct sum $L=\sum_i L_i$. Then, for any $b$, 
 \[
 \Pr_{a \in [n]^d}[f(b) = L(b)] \geq \frac{3}{4}.
 \]
It follows that $f_{{\sf shapka}} = L$.
\end{proposition}
\begin{proof}
We say that a query $f(x)$
\defn{succeeds} if $f(x)=L(x)$.
Let $\eps_i = \Pr_x[f(x)=L(x)\mid x_i = b_i]$.
By a union bound, the Shapka test succeeds with probability at least $1-\sum_i \eps_i$.
For each $i$, 
$$\eps=\Pr[f(x)\neq L(x)]\ge \Pr[f(x)\neq L(x)\land x_i = b_i]=\eps_i/n.$$
Thus, $\sum_i \eps_i \le n\eps d \le 1/4,$ so the Shapka test produces the correct value with probability at least $3/4$.
\end{proof}

Now we give a more query efficient method for reconstructing a direct sum. Our method has asymptotically optimal query complexity (by \cref{prop:querylb}),
and still works at essentially all possible values of $\eps$ for which the
function is guaranteed to be information theoretically recoverable (by \cref{prop:infolb}). We reconstruct the direct sum via the following voting scheme, where we call the reconstructed function $f_{{\sf recon}}$. To make things more convenient, we assume that $d$ is odd. Set 
\[
f_{{\sf recon}}(x) = \maj_{x^{(1)}, \ldots, x^{(n)}}\sum_{i = 1}^n x^{(i)},
\]
where the set of queried points $x^{(1)}, \ldots, x^{(n)}$ is weighted in the majority according to the following distribution.

\vspace{0.1cm}

\begin{itemize}
    \item Randomly partition  $[d]$ into $n$ parts, $S_1 \sqcup \cdots \sqcup S_n = [d]$, by choosing putting each $i\in [d]$ independently and uniformly at random in a part $S_i$.
    \item Sample $R$  uniformly from $([n]\setminus\set{1})^{d}$.
    \item For $i\in [n]$, form query point $x^{(i)}\in [n]^d$ by setting $x^{(i)}_j=1$ for each $j\in S_i$, and setting $x^{(i)}_j = R_j$ otherwise.
\end{itemize}

For even, $n$, we require $n+1$ queries, and use the following voting scheme. 

\[
f_{{\sf recon}}(x) = \maj_{x^{(1)}, \ldots, x^{(n)}}\sum_{i = 1}^n x^{(i)},
\]
where the set of queried points $x^{(1)}, \ldots, x^{(n+1)}$ is weighted in the majority according to the following distribution.

\vspace{0.1cm}

\begin{itemize}
    \item Randomly partition  $[d]$ into $n+1$ parts, $S_1 \sqcup \cdots \sqcup S_{n+1} = [d]$, by choosing putting each $i\in [d]$ independently and uniformly at random in a part $S_i$.
    \item Sample $R\in [n]^d$ as follows: for each $i$  independently, set $R_i=1$ with probability $1/n^2$ and otherwise sample $R_i$ uniformly from $[n]\setminus \set{1}$.
    \item For $i\in [n+1]$, form query point $x^{(i)}\in [n]^d$ by setting $x^{(i)}_j=1$ for each $j\in S_i$, and setting $x^{(i)}_j = R_j$ otherwise.
\end{itemize}
We prove that the scheme above correctly reconstructs the underlying direct sum for the case where $n$ is odd, as the case where $n$ is even is similar.
\begin{proposition}
Fix $n,d,\eps$. Let $f: [n]^d \to \F_2$ be $\eps$ close to a direct sum $L$,
with $(n+1)\eps < 1/4$. Then for any point $x$
\[
\Pr_{x^{(1)},\ldots, x^{(n)}}\left[\sum_{i=1}^n f(x^{(i)}) = L(x)\right] \geq \frac{3}{4}.
\]
It follows that $f_{{\sf recon}} = L$.
\end{proposition}
\begin{proof}
 Without loss of generality, let us assume that $x$ is the all ones point, $\mathbb{1}$.  We again say that a query $f(x)$ \defn{succeeds} if $f(x)=L(x)$.

We claim that $x^{(1)},\ldots, x^{(n)}$ are uniformly random points in $[n]^d$. 
Fix $i\in [n]$. For each $j\in [d]$,
$x^{(i)}_j$ has a $1/n$ chance of being a $1$, and a $1/n$ chance of being $j$ for
any $j\in [n]\setminus\set{1}$. Furthermore, the value of $x^{(i)}_j$ is
independent of the value of $x^{(i)}_{j'}$ for $j'\neq j$.
Thus, the probability that all $n$ queries all succeed is at least $1-n\eps \ge 3/4$ by a union bound.
If the queries all succeed, then because $d$ is odd, we have:
$$\sum_i f(x^{(i)}) = \sum_i L(x^{(i)})  = (d-1)L(R) + L(\one) = L(\one).$$
% We claim that $x^{(1)},\ldots, x^{(n+1)}$ are uniformly random points in $[n]^d$. 
% Indeed, fixing $i\in [n+1],j\in [d]$ we have that 
% $$\Pr[x^{(i)}_j=1] = \frac{1}{n+1} + \frac{n}{n+1} \frac{1}{n^2} =\frac{1}{n},$$
% and $x^{(i)}_j$ takes on all other values with probability $1/n$ as well, and $x^{(i)}_j$ is independent of $x^{(i)}_{j'}$ for $j'\neq j$.
% Thus, by a union bound all queries simultaneously succeed with probability at least $3/4$.
% In this case, adding $f(x^{(i)})$ at the queried points yields $L(\one)$.
\end{proof}